\documentclass[11pt,a4paper]{article}
\usepackage[utf8]{inputenc}
\usepackage[T1]{fontenc}
\usepackage{tgtermes}
\usepackage{amsmath,amsthm,amssymb}
\usepackage[margin=1.05in]{geometry}
\usepackage{setspace}
\usepackage{booktabs}
\usepackage{array}
\usepackage{graphicx}
\usepackage[hidelinks]{hyperref}
\theoremstyle{plain}
\newtheorem{theorem}{Theorem}
\newtheorem{proposition}{Proposition}
\newtheorem{lemma}{Lemma}
\newtheorem{corollary}{Corollary}
\theoremstyle{definition}
\newtheorem{definition}{Definition}
\theoremstyle{remark}
\newtheorem{remark}{Remark}

\newcounter{aksc}
\newenvironment{assumption}[1]{\refstepcounter{aksc}\par\medskip\noindent\textbf{Assumption A\arabic{aksc} (#1).}\itshape\ }{\par\medskip}

\newcommand{\E}{\mathbb{E}}
\newcommand{\Pj}{\Pr}

\title{\textbf{Crisis without Collapse:\\ A Theory of Capitalist Persistence}}
\author{Murat Ova\thanks{Independent researcher; murat@stochasticminds.com. I thank Georgy Lukyanov for detailed comments that substantially improved the paper. All remaining errors are my own. Comments welcome.}}
\date{August 2026}

\begin{document}
\maketitle
\vspace{-3.2em}

\begin{abstract}
\noindent The terminal crisis of capitalism has been predicted regularly since 1848. It has not occurred. We treat this 175-year forecasting record as a stylized fact in need of a model, and we supply one. The model unifies results from five literatures: adaptive systems, mechanism design and computational complexity, behavioral decision theory, global games, and the political economy of redistribution. Ten assumptions deliver twelve formal results and one theorem. Recurring shock classes have geometrically decaying collapse hazards, and the institutional response lag within the banking class alone has fallen from 816 days (1930--33) to 18 days (2008) to 2 days (2023). A centralized replacement of the price mechanism fails on at least one of two horns, computation or truthful input, and we are explicit about which horn survives modern algorithms. Revolutionary participation requires a reference-point break, and we derive Davies's J-curve from cumulative prospect theory together with a new result: the revolutionary window closes endogenously, with length inversely proportional to the speed of hedonic adaptation, roughly four to nine years at calibrated parameters, and with its perceptibility threshold derived from the participation constraint rather than assumed. Personalized media shift protest games, in the model, into the unique no-attack equilibrium of a global game. Redistribution is a dial, not an identity: countries with nearly identical market Ginis (0.505--0.512) sustain disposable Ginis from 0.273 to 0.395. We state seven testable predictions and six falsification conditions, and we map which assumption each rival school must deny.\\[4pt]
\noindent\textbf{Keywords:} institutional persistence; socialist calculation debate; prospect theory; global games; soft budget constraint; political economy of redistribution.\\
\noindent\textbf{JEL codes:} B24, C73, D02, D82, D83, D91, P16, P21.
\end{abstract}

\newpage
{\small \tableofcontents}
\newpage

\section{Introduction}\label{sec:intro}

The terminal crisis of capitalism has been announced, by our count, at least fifteen times since 1848.\footnote{The count, generous to the announcers: 1848, 1857, 1873, 1893, 1907, 1914--18, 1929--33, 1938, 1968--71, 1973--75, 1987, 1997--98, 2000, 2008, 2020, and 2023. Table~\ref{tab:record} details the ten with the largest realized disruptions; the binomial bound in Section~\ref{sec:record} uses fifteen, which is the conservative direction.} Table~\ref{tab:record} collects the major announcements and the outcomes. The pattern is awkward for the announcers. Every episode ends the same way: a recession, a reform, and a larger rulebook. A directional forecast that has failed for 175 consecutive years is not, by the standards of social science, an embarrassment to be explained away. It is a stylized fact, and stylized facts deserve models.

This paper supplies the model. We take the persistence of the capitalist system as the object to be explained, in the same spirit in which growth theory takes the Kaldor facts, and we ask what set of primitive assumptions generates it. Our answer assembles results that already exist, scattered across five literatures that rarely cite one another: the economics of adaptive systems, the mechanism-design reading of the socialist calculation debate, behavioral decision theory, the theory of coordination under incomplete information, and the political economy of redistribution. Individually, most of our propositions are known, or are near neighbors of known results. The contribution is the assembly: ten assumptions, stated once, from which the entire structure follows, together with an explicit account of where the structure can break. The paper is, by design, a synthesis with formal scaffolding rather than a unified model, and it says so where it matters (Remark~\ref{r:honest}). Its most portable product may well be Table~\ref{tab:rivals}, which converts a century-old war of quotations into a schedule of measurements, usable even by readers who reject our conclusion.

Four mechanisms carry the argument. First, the system learns. Institutional responses to recurring shock classes arrive faster and bind harder with every repetition, so the collapse hazard of any known class decays geometrically (Proposition~\ref{p:immune}). The point is not rhetorical. Within the banking-panic class alone, the lag between systemic bank failure and a system-wide deposit backstop fell from 816 days in 1930--33 to 18 days in 2008 to 2 days in 2023 (Figure~\ref{fig:record}b). Second, prices are a communication technology with no known substitute of equal bandwidth. We restate the calculation debate in its modern form, as a two-horn dilemma over computation and truthful input, and we are careful about which horn survives improvements in algorithms and hardware (Proposition~\ref{p:dilemma}). The computation horn is weaker than its popularizers believe. The input horn is not.

Third, and this is where the paper is closest to being new, the revolutionary agent is missing for reasons that can be derived rather than asserted. Under cumulative prospect theory with an adaptive reference point, participation in revolt is possible only inside windows opened by a sharp break in living standards. We derive Davies's (1962) J-curve as a proposition rather than a hypothesis, and we obtain a result we have not seen elsewhere in this form: the window closes on its own, with length
\[
k^{*} \;\approx\; \frac{1}{\gamma}\,\ln\frac{\Delta}{\varepsilon},
\]
inversely proportional to the speed $\gamma$ of hedonic adaptation, and with the threshold $\varepsilon$ derived from the participation constraint rather than assumed (Proposition~\ref{p:jwindow}). At adaptation half-lives of 1.4 to 3.1 years, the values implied by the happiness-adaptation literature, revolutionary windows last roughly four to nine years. The historical clustering of successful revolutions within half a decade of their precipitating shocks, 1789, 1917, 1979, 1989, is what the arithmetic predicts. So is the folk observation, common to every failed revolutionary generation, that the moment was missed.

Fourth, even inside an open window, mass action is a coordination game, and the equilibrium that gets played depends on the information architecture. City squares, single-channel broadcasting, and Friday prayer networks are public signals: everyone sees them and everyone sees that everyone sees them. Personalized feeds are not. In the global-games benchmark, replacing public with private signals collapses the multiplicity that revolution requires into a unique equilibrium in which attack occurs only at extreme fundamentals (Proposition~\ref{p:coord}). The atomization of attention is not a conspiracy and needs none; it is an equilibrium information structure, and it is one in which a million people can hold the same grievance separately, forever.

A fifth mechanism disciplines the other four. Distribution is a policy dial, not a system identity. In 2019 OECD data, the United States, the United Kingdom, and Finland had market-income Ginis of 0.505, 0.508, and 0.512, and disposable-income Ginis of 0.395, 0.366, and 0.273 (Figure~\ref{fig:gini}). The market near-identity is partly composition, pensions and employment structure, and we do not lean on it; the claim that carries weight is the dial. The dial has a 12-point range within a single institutional family, and the elite's equilibrium strategy is to turn it exactly as far as the revolution constraint requires (Proposition~\ref{p:dial}). Bismarck understood this in the 1880s. The formal version is Acemoglu and Robinson (2000, 2006), and the empirical version is Aidt and Jensen (2014).

We combine the mechanisms into a single persistence theorem (Theorem~\ref{t:main}) and then spend a full section on the ways it can fail. The theorem is exactly as strong as its assumptions, and we say so in the statement rather than the footnotes. The assumption we regard as the strongest empirical bet is the meta-learning property A2, the claim that institutional adaptation transfers across shock classes fast enough to keep the hazard of \emph{novel} shocks bounded. Climate tipping dynamics are the natural candidate violation, and they head our falsification list. We also flag, against our own side, the creeping generalization of too-big-to-fail as an import of Kornai's soft budget constraint into the heart of the system it originally described from outside.

\paragraph{What this paper is not.} The theorem is positive, not normative. Persistent does not mean good, and nothing below should be read as a welfare statement; Proposition~\ref{p:status} and Corollary~\ref{c:piketty} record, within the model, the distributional tensions the system persists \emph{with}. Nor is the paper a forecast that nothing can change. It is a statement of the conditions under which change of a particular kind, the permanent replacement of the price mechanism, can occur, and a demonstration that those conditions are jointly very restrictive. Nor does the paper treat capitalism and the price mechanism as synonyms; the theorem's object is the coordination core, and Remark~\ref{r:op} states exactly what is and is not claimed about ownership.

\paragraph{Related literature.} The adaptive reading of market economies goes back to Alchian (1950) and is institutionalized by Nelson and Winter (1982); the option-theoretic account of why volatility feeds the system draws on Taleb and Douady (2013), and the selection-by-crisis mechanism on Schumpeter (1942). The calculation debate begins with Mises (1920) and Hayek (1945); its mechanism-design form is Hurwicz (1972, 1977), Mount and Reiter (1974), and Jordan (1982), and its incentive form is Myerson and Satterthwaite (1983). The complexity reading owes much to Shalizi's (2012) essay, which we sharpen and, on one point, correct: modern linear-programming bounds (Cohen, Lee, and Song, 2019) weaken the naive operation-count argument, and the load-bearing wall is nonconvexity (Karp, 1972) together with the input horn. Kornai (1980, 1986) and Dewatripont and Maskin (1995) supply the soft budget constraint; Barzel (1974) prices the queue. The behavioral block is Kahneman and Tversky (1979) and Tversky and Kahneman (1992), with adaptation dynamics from the happiness literature (Brickman, Coates, and Janoff-Bulman, 1978; Clark, Frijters, and Shields, 2008; Di Tella, Haisken-De New, and MacCulloch, 2010). Prospect-theoretic accounts of rebellion exist (Berejikian, 1992; Levy, 1992); the endogenous closing of the window, to our knowledge, does not. The coordination block is Aumann (1976), Rubinstein (1989), Carlsson and van Damme (1993), Morris and Shin (1998, 2002, 2003), Hellwig (2002), Angeletos and Werning (2006), and Edmond (2013), with the sociology supplied by Granovetter (1978), Kuran (1989, 1991), Chwe (2001), and Lohmann (1994). The political block is Acemoglu and Robinson (2000, 2006), tested by Aidt and Jensen (2014), and disciplined from the left by Piketty (2014). Rival programs we engage at the assumption level include Lange (1936--37), Cockshott and Cottrell (1993), Shaikh (2016), and the degrowth school (Kallis, 2018; Hickel, 2020).

The paper proceeds as follows. Section~\ref{sec:record} documents the record. Section~\ref{sec:model} lays out the environment and the ten assumptions. Sections~\ref{sec:adaptive} through~\ref{sec:reinforce} derive the results, one mechanism per section. Section~\ref{sec:theorem} states the theorem. Section~\ref{sec:predictions} lists predictions, falsification conditions, and the map of rival positions. Section~\ref{sec:conclusion} concludes. Proofs are in Appendix~\ref{app:proofs}; data construction is in Appendix~\ref{app:data}.

\section{The Record}\label{sec:record}

\begin{table}[t]
\centering
\small
\begin{tabular}{p{1.5cm} p{4.6cm} p{4.2cm} p{3.6cm}}
\toprule
\textbf{Episode} & \textbf{Contemporary verdict} & \textbf{Peak disruption} & \textbf{Institutional residue} \\
\midrule
1873 & ``the crisis of capitalism'' & 65-month contraction & antitrust era \\
1929--33 & Comintern: the final crisis & 43 months; U.S. output falls by roughly a quarter & FDIC, Glass--Steagall, SEC \\
1938 & secular stagnation & recovery interrupted by war & Bretton Woods (1944) \\
1971 & end of the dollar order & stagflation decade & floating rates; later, inflation targeting \\
1987 & fastest crash on record & $-22.6$\% in one day & circuit breakers \\
1997--98 & global contagion & regional depressions & swap-line architecture \\
2000 & the new economy is dead & Nasdaq $-78$\% peak to trough & Sarbanes--Oxley \\
2008 & ``the end of capitalism'' & 18-month contraction & Dodd--Frank, Basel III, stress tests \\
2020 & fastest 30\% drawdown ever & 2-month contraction & pandemic facility playbook \\
2023 & ``GFC 2.0'' & three of the four largest U.S. bank failures & BTFP template \\
\bottomrule
\end{tabular}
\caption{\small Ten announced terminal crises and their outcomes. Contraction lengths follow the NBER business-cycle chronology; verdicts paraphrase contemporary commentary. Every episode ends in adaptation, none in exit. Sources and definitions: Appendix~\ref{app:data}.}
\label{tab:record}
\end{table}

Two features of the record matter for what follows, and Figure~\ref{fig:record} displays both. First, crises are not lengthening. Panel (a) plots the duration of every U.S. contraction in the NBER chronology since 1873. The three longest contractions on record are the three earliest major ones; the fitted trend is negative. The claim that each crisis is deeper than the last, which is load-bearing in the cumulative-contradiction tradition, is not visible in a century and a half of business-cycle dating.

Second, within a fixed shock class, the institutional response accelerates by orders of magnitude. Panel (b) takes the banking-panic class, the oldest and best-documented one. In the 1930--33 wave, 816 days separated the onset of systemic bank failures from a system-wide deposit backstop, and thousands of banks failed in the interval. In 2008 the analogous lag was 18 days, with the money-market guarantee arriving in 4. In 2023 it was a weekend: Silicon Valley Bank was closed on a Friday and the systemic-risk exception plus the Bank Term Funding Program were announced on Sunday evening, before Asian markets opened. No depositor lost a dollar. Same pathogen, three exposures, and a response time falling by a factor of forty at each exposure. This is what Assumption A1 will call immune memory, and it is measured, not asserted.

Three caveats keep the series honest. First, the denominator moved: deposits left Silicon Valley Bank at a pace physically impossible in 1930, so part of the fall in the lag is a race against accelerating shock propagation, and part measures the diffusion of an invention, since the instrument being timed did not exist in 1930. The discriminating observable is severity: if the series recorded fragility outrunning response, damage conditional on the shock would not fall, and it has fallen to zero, from thousands of failed banks and roughly \$1.3 billion of depositor losses in 1930--33 to none in 2023. Second, the frequency margin cuts the other way, and we do not hide it: banking-crisis incidence in advanced economies has been higher after 1973 than under Bretton Woods (Schularick and Taylor, 2012; Reinhart and Rogoff, 2009), an era whose calm was purchased by financial repression, itself a dial turned hard on finance. Lethality per hit has fallen while hits per decade have risen; the quantity that matters is the product, we have not computed it here, and its adverse movement is precisely what F5 is built to register, with A2 the exposed assumption if the rise is driven by mutated classes. Third, the NBER chronology before 1927 is retrospective Burns--Mitchell dating, not constructed comparably to the modern committee's, and the negative trend in panel (a) leans on 1873 and 1929; we therefore read the panel as the absence of a rising trend, which survives the caveat, rather than as a measured decline, which might not.

\begin{figure}[t]
\centering
\includegraphics[width=\textwidth]{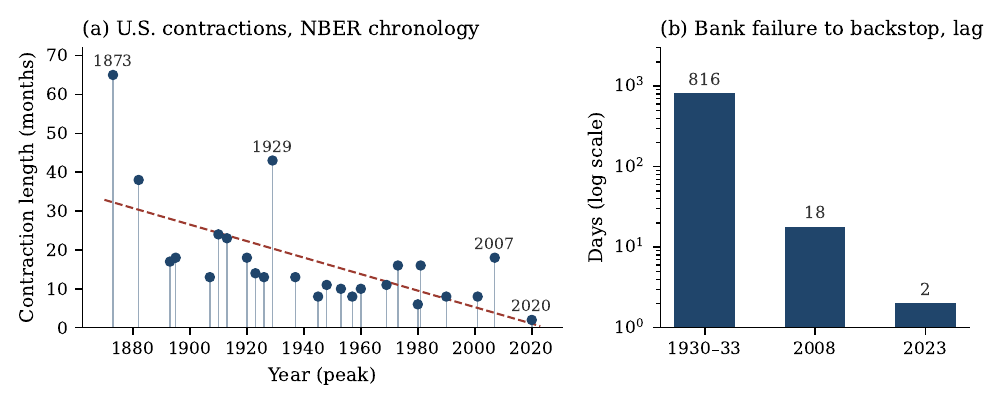}
\caption{\small \textbf{The record.} (a) Duration of U.S. business-cycle contractions, NBER chronology, 1873--2020, with linear fit (dashed) shown for orientation; the pre-1927 dating caveat of Section~\ref{sec:record} applies, and we read the panel as the absence of a rising trend rather than as a measured decline. (b) Days from onset of systemic bank failures to announcement of a system-wide deposit backstop, log scale: 1930--33 (Bank of United States failure to the national bank holiday and Emergency Banking Act), 2008 (Lehman to EESA; the money-market guarantee came in 4 days), 2023 (SVB closure to the systemic-risk exception and BTFP). Construction details in Appendix~\ref{app:data}.}
\label{fig:record}
\end{figure}

We do not claim the record proves persistence; fifteen episodes cannot prove a probability zero. A crude calculation is still instructive. If each announced terminal episode carried an independent collapse probability $q$, then observing zero collapses in fifteen trials bounds $q$ below $1-0.05^{1/15}\approx 0.18$ at the 95 percent level. Note what the bound does not say: $q$ as high as 0.18, roughly one chance in six per announced crisis, is fully consistent with this record. Read as a doom estimate rather than a reassurance, the number is livelier than the persistence rhetoric suggests, which is exactly why the theorem leans on mechanisms rather than on this count. Independence is, moreover, generous to the collapse hypothesis: under Proposition~\ref{p:immune} the trials are negatively dependent, each failure lowering the hazard of the next, so the effective bound tightens over time. Two caveats bind, and we state them rather than bury them. The sample of announcements is curated, not a census; and the negative-dependence defense leans on Proposition~\ref{p:immune}, that is, on the model whose plausibility is at issue, so the calculation illustrates magnitudes rather than supplying independent evidence. What the model itself asserts, through Proposition~\ref{p:prophecy}, is that each announcement feeds the adaptation that refutes it; on that reading the record is the mechanism's footprint, not an outside test, and that is how we use it.

\section{Environment and Assumptions}\label{sec:model}

\subsection{Primitives}

Time is discrete, $t=0,1,2,\dots$ The economy at date $t$ is a pair $E_t=(x_t,R_t)$, where $x_t\in\mathcal{X}$ is the \emph{state} (allocations, the technology stock, the production network, the price vector) and $R_t\in\mathcal{R}$ is the \emph{rulebook} (deposit insurance, capital requirements, competition law, the tax-and-transfer schedule). Shocks $\{\sigma_t\}$ arrive over time and are partitioned into \emph{classes}, $\Sigma=\biguplus_k \Sigma_k$: bank runs are one class, supply disruptions another, pandemics a third. $N_k(t)$ counts realized hits of class $k$ through date $t$.

\begin{definition}[Viability, crisis, collapse]\label{d:collapse}
Let $\mathcal{V}\subset\mathcal{X}$ be the set of states in which the production--exchange--finance loop reproduces itself, and $\mathcal{V}_\varepsilon\supset\mathcal{V}$ a stressed neighborhood. A \emph{crisis} is an excursion of $x_t$ out of $\mathcal{V}$ with finite return time. A \emph{near miss} is an excursion beyond $\mathcal{V}_\varepsilon$ that still returns. \emph{Collapse} is absorbing exit, and it comes in exactly two forms:
\begin{itemize}
\item[(E1)] \emph{Disintegration}: $x_t\notin\mathcal{V}$ for all $t\geq T$, with no functioning coordination mechanism replacing the old one.
\item[(E2)] \emph{Replacement}: the price mechanism is permanently replaced by a non-price coordination mechanism.
\end{itemize}
\end{definition}

\begin{remark}[Operationalizing the definitions]\label{r:op}
Three clarifications, all prompted by the Soviet case. First, by \emph{the price mechanism} we mean decentralized, market-clearing prices formed in exchange under binding budget constraints. Administered price tags of the Gosplan type are not prices in this sense: they cleared no markets and carried no scarcity information, which is why the Soviet core is coded here as a non-price regime despite its rubles and its price lists. Second, permanence is never observable in finite time, so E2 needs an operational horizon, and honesty requires naming the framework's nearest miss. The Soviet core ran sixty-two years from the abolition of NEP in 1929 to 1991, the longest non-price trial on record; it ended in disintegration, while the socialist economies that survived did so by reinstalling the price module (Section~\ref{sec:reinforce}). We concede without embarrassment that an observer writing in 1980 could reasonably have scored that trial as a refutation in progress; the framework survives it on the record's eventual verdict, not on evidence available at the time. For falsification going forward we therefore fix the horizon in advance: a non-price replacement that runs a century without reinstalling prices and without disintegrating counts against Theorem~\ref{t:main}(i). The century is admittedly post hoc relative to the Soviet trial; pre-registering it now binds us for the next one. Third, the title says capitalism and the theorem says price mechanism, and the two are not synonyms. Ownership arrangements lie outside the theorem's object: Lange-type market socialism, cooperative ownership, and state-heavy mixed economies all retain the price core and are therefore not E2 events, so the theorem is silent about them, and the China case supports the price-module claim rather than any claim about private ownership. We keep the word capitalism because Table~\ref{tab:record} is the record of the historical family organized around the price core; readers who mentally retitle the paper ``The Persistence of the Price Mechanism'' are reading the theorem exactly correctly.
\end{remark}

Every episode in Table~\ref{tab:record} is a crisis or a near miss under this definition; the set of realized E1 or E2 events in the observation window is empty. After each shock the rulebook updates through an \emph{adaptation operator}, $R_{t+1}=A(R_t,\sigma_t,y_t)$, where $y_t$ is observed damage. Let $q_k^{(m)}$ denote the probability that the $m$-th hit of class $k$ causes collapse given the rulebook in force at the hit, and $p_k^{(m)}$ the corresponding near-miss probability. Finally, a path $x_0\to x_1\to\dots\to x_T$ is \emph{feasible} if every intermediate state lies in $\mathcal{V}$ and every step is a local module change. Heart surgery on this patient is performed while the patient runs.

\subsection{Assumptions}

We group the assumptions in four blocks: structural (A1--A3), informational (A4--A5), behavioral (A6--A7), and political (A8--A10). Each carries its empirical anchor in brackets. A2 is the strongest empirical bet in the system, and we flag it as such. A note on labels before the results: propositions whose statements carry citations are imports from the literature or direct applications of it, numbered as propositions for cross-reference convenience; the paper claims originality only for the threshold and closing results of Proposition~\ref{p:jwindow}, the open-window share, the assembly itself, and the falsification mapping.

\begin{assumption}{Immune memory}\label{a:immune}
There exists $\delta\in(0,1)$ such that for every class $k$ and hit count $m$, $q_k^{(m+1)}\leq\delta\,q_k^{(m)}$ and $p_k^{(m+1)}\leq\delta\,p_k^{(m)}$. {\upshape [Anchor: Figure~\ref{fig:record}b. The banking class went from thousands of failures and large depositor losses in 1930--33 to three large failures and zero depositor losses in 2023.]}
\end{assumption}

\begin{assumption}{Bounded novelty and meta-learning}\label{a:novelty}
First-hit hazards of novel classes are uniformly bounded, $\sup_k q_k^{(1)}\leq\bar q<1$, and adaptation transfers across classes: the bound $\bar q_t$ is non-increasing in $t$. {\upshape [Anchor: the 2020 pandemic was a novel class, and the 2008 toolkit, swap lines, asset purchases, rapid fiscal support, was redeployed against it within weeks. This is the assumption climate dynamics will test; see F1.]}
\end{assumption}

\begin{assumption}{Option structure}\label{a:option}
The system comprises $n$ separable units with limited liability and first-order shock exposure, $V_i(x+\epsilon)\approx V_i(x)+\beta_i^{\top}\epsilon$, and aggregates as $F(x)=\sum_i \max\{V_i(x),0\}$. {\upshape [Anchor: bankruptcy law; Schumpeter (1942). The 2000 crash destroyed the dot-com sector and bequeathed its fiber and engineers to the next round.]}
\end{assumption}

\begin{assumption}{Information primitives}\label{a:info}
Production and consumption knowledge is (i) dispersed, $I=\bigcup_j I_j$ with no $j$ holding $I$; (ii) partly tacit, readable only from choice under a binding budget constraint, not from reports; (iii) local and time-varying. {\upshape [Anchor: Hayek (1945); the systematic gap between survey statements and till receipts.]}
\end{assumption}

\begin{assumption}{Structure of the planning problem}\label{a:plan}
The centralized allocation problem has $N$ variables, with $N$ scaling as (commodities)$\times$(locations)$\times$(dates)$\times$(contingencies) and $N\geq 10^{10}$ in realistic economies; it is generically nonconvex, because of indivisibilities, fixed costs, and increasing returns; and its coefficient data go stale over a horizon $\tau_{s}$. {\upshape [Anchor: Gosplan planned over roughly twelve million line items annually (Nove, 1983); scale economies are the defining feature of the industries that matter.]}
\end{assumption}

\begin{assumption}{Behavioral core}\label{a:cpt}
Agents code outcomes relative to a reference point and evaluate them with the cumulative-prospect value function
\[
v(z)=\begin{cases} z^{\alpha}, & z\geq 0,\\ -\lambda(-z)^{\beta}, & z<0,\end{cases}
\qquad \alpha\approx\beta\approx 0.88,\quad \lambda\approx 2.25,
\]
overweight small probabilities, $w(p)>p$ for small $p$, and discount quasi-hyperbolically. {\upshape [Anchor: Kahneman and Tversky (1979); Tversky and Kahneman (1992); Thaler (1981); Laibson (1997). Risk seeking over losses holds for moderate-to-high probabilities, the loss branch of the fourfold pattern; Proposition~\ref{p:jwindow} uses the curvature of that branch, diminishing sensitivity, rather than the risk attitude itself.]}
\end{assumption}

\begin{assumption}{Adaptive reference}\label{a:ref}
The reference point tracks realized consumption geometrically: $r_{t+1}=(1-\gamma)r_t+\gamma c_t$, $\gamma\in(0,1)$. {\upshape [Anchor: hedonic adaptation; Brickman, Coates, and Janoff-Bulman (1978); Clark, Frijters, and Shields (2008). Panel estimates imply most of the well-being effect of an income change dissipates within roughly four years (Di Tella, Haisken-De New, and MacCulloch, 2010), i.e.\ $\gamma$ in the range 0.2--0.4.]}
\end{assumption}

\begin{assumption}{Information structure of the regime-change game}\label{a:regime}
Revolution is a threshold game with strategic complementarities: the regime falls if the attacking mass exceeds $\theta(f)$ at fundamental $f$. Signals are of two kinds, \emph{public} (observed by all, and observed by all to be observed by all) and \emph{private} (noisy, with privately known precision). In the personalized-media regime there is no public signal correlated with aggregate discontent. {\upshape [Anchor: the passage from the city square and single-channel broadcasting to the individually curated feed; Chwe (2001).]}
\end{assumption}

\begin{assumption}{Dependency network}\label{a:network}
The state $x$ is a densely connected network of modules; global reconfiguration is possible only along feasible paths in the sense above, and halting the system for a rebuild is immediate exit from $\mathcal{V}$. {\upshape [Anchor: a single canal blockage in 2021 held up roughly \$9.6 billion of trade flow per day; advanced logic capacity is concentrated on one island.]}
\end{assumption}

\begin{assumption}{Control loop}\label{a:dial}
The distribution parameter $\theta\in\Theta$ (the intensity of taxes and transfers) is adjustable by the political system, and the system's identity does not depend on its value: the market economy remains in $\mathcal{V}$ over a wide range of $\theta$. {\upshape [Anchor: Figure~\ref{fig:gini}; Bismarck's social-insurance package of the 1880s.]}
\end{assumption}

\section{The Adaptive Layer}\label{sec:adaptive}

\begin{proposition}[Immune memory]\label{p:immune}
Under A1: (i) for every known class $k$, $\sum_m q_k^{(m)}\leq q_k^{(1)}/(1-\delta)<\infty$; (ii) by countable subadditivity, $\Pj[\text{collapse from known classes}]\leq\sum_{k\in K} q_k^{(1)}/(1-\delta)$; (iii) by the first Borel--Cantelli lemma, each class generates almost surely only finitely many near misses.
\end{proposition}

The union bound in (ii) is the right tool because collapse is absorbing; Borel--Cantelli, which requires no independence in its first half, is the right tool for the recurring near misses. The economics is the essay's phrase made precise: the same virus cannot make the system equally sick twice. What decays is not the arrival rate of shocks but their conditional lethality.

\begin{proposition}[Antifragility of the aggregate]\label{p:anti}
Under A3, $F$ is convex, and for any zero-mean shock $\epsilon$, $\E[F(x+\epsilon)]\geq F(x)$, strictly if the shock distribution puts mass on both sides of at least one unit's kink. The aggregate gains from volatility that destroys its parts.
\end{proposition}

The mechanism is Jensen's inequality applied to a portfolio of limited-liability claims, the convexity definition of antifragility in Taleb and Douady (2013), and the selection reading is Alchian (1950): volatility is the filter that reallocates share from maladapted units to adapted ones. Nothing here requires anyone to intend it.

\begin{proposition}[Prophecy inversion]\label{p:prophecy}
Let the public collapse belief be $q$ and the realized hazard $q'=g(q)$. If institutional response makes $g$ continuous and strictly decreasing, then $g$ has a unique fixed point $q^{*}$; if in addition $\sup|g'|<1$, the belief dynamic $q_{n+1}=g(q_n)$ converges to it. Beliefs about collapse are then self-defeating rather than self-fulfilling. (Without the contraction bound a decreasing map can settle into a two-cycle: uniqueness of the fixed point is general, convergence is not.)
\end{proposition}

Merton (1948) described the increasing case. Institutional design is the engineering of the sign flip, and the canonical worked example is Diamond and Dybvig (1983): the bank run is a coordination game with two equilibria, and deposit insurance deletes the bad one from the game tree. Two consequences deserve emphasis. First, the loudest doomsayers are, in this mechanism, unpaid sensors: every widely circulated collapse scenario raises $q$, which triggers the response that lowers $q'$. The system's most faithful debuggers are its critics; the mirror-image worry, that sustained success erodes the very supply of credible critics, is developed in Lukyanov (2026) and becomes our sixth falsification condition. Second, the mechanism explains why the forecasting record of Table~\ref{tab:record} \emph{worsens over time for the forecasters}: each announcement contributes to its own refutation.

\section{The Informational Layer}\label{sec:info}

\begin{lemma}[Informational minimality of prices]\label{l:prices}
In classical exchange environments (convex preferences and technology, no externalities), within the class of regular privacy-preserving mechanisms realizing Pareto-optimal and individually rational allocations, the competitive price mechanism uses a message space of minimal dimension, and it is the unique such mechanism (Mount and Reiter, 1974; Hurwicz, 1977; Jordan, 1982).
\end{lemma}

We import the lemma rather than prove it, and the scope conditions matter: convexity and no externalities. Its role here is narrow but structural. The familiar metaphor of the price as a compressed telegram, scarcity plus demand plus expectation in one number, is not a metaphor. It is a dimension-counting theorem. Where externalities are present the minimality fails, which is precisely the gap Section~\ref{sec:reinforce} works through for carbon.

\begin{proposition}[The planning dilemma]\label{p:dilemma}
Under A4 and A5, any centralized replacement of the price mechanism faces at least one of two horns at all times:
\begin{itemize}
\item[\emph{H1}] (\emph{Computation.}) The realistic planning problem cannot be solved exactly within the staleness horizon $\tau_{s}$ under known algorithms and standard complexity assumptions ($\mathrm{P}\neq\mathrm{NP}$).
\item[\emph{H2}] (\emph{Input.}) The objective function's inputs, true preferences and costs, cannot be elicited without binding budget constraints, that is, without prices.
\end{itemize}
\end{proposition}

H2 carries the argument; H1 we state mainly to concede its popular version, and the concession is worth making precisely.

First, then, H1. The classical interior-point bound for linear programs is $O(N^{3.5})$ operations; the best known theoretical bound is now roughly $O(N^{2.37})$ (Cohen, Lee, and Song, 2019), and Figure~\ref{fig:comp} plots both at exascale. At $N=10^{10}$ the optimistic exponent gives about six days per deterministic solve: painful, not impossible. What remains after the concession is not an exponent. Contingent commodities multiply $N$ combinatorially in the number of distinguishable states, and A5's nonconvexities, the fixed costs and increasing returns of exactly the industries a plan would most need to get right, take the problem out of the convex class altogether: integer programming is NP-complete (Karp, 1972).

Second, honesty requires the converse observation. Markets do not solve this optimization either. They do not need to. The market implements a massively parallel, locally informed, adaptive approximation in which no node ever holds the objective function. The correct comparison is not optimal plan versus market; it is uncomputable optimum versus computable approximation. H1 does not say the plan would be worse than the market at the market's own game. It says the game the plan proposes to play is not defined.

Third, H2 is the binding horn, and it survives arbitrary hardware. With costless reports, the mechanism is a cheap-talk game, and when allocations respond to reports, exaggeration is a best reply: informative separation collapses into babbling. Attaching a cost to reports reinvents the price. The impossibility anchors are general. In classical environments no individually rational, Pareto-efficient mechanism is dominant-strategy incentive compatible (Hurwicz, 1972), and even in the minimal case of bilateral trade, incentive compatibility, individual rationality, budget balance, and full efficiency cannot coexist (Myerson and Satterthwaite, 1983). A4(ii) deepens the point: the tacit component is information the agent herself does not possess before choosing under a binding constraint. A survey is the wrong instrument by construction, not by execution.

\begin{remark}[The AI planner and self-reference]\label{r:ai}
The objection that sufficiently capable machine intelligence dissolves H1 runs into H2. A planning model's training distribution is the record of choices under binding budget constraints, that is, of market interactions. Retiring the price mechanism switches off the process that generates the model's own inputs; the system becomes a loop consuming its own tail. We state this as an argument, not a theorem, and we give it a falsification condition (F2): a demonstrated preference-revelation technology that works without binding constraints would defeat it, and would be, in our view, the single most important discovery a planning research program could make. Solving H1 alone would not.
\end{remark}

\begin{figure}[t]
\centering
\includegraphics[width=0.72\textwidth]{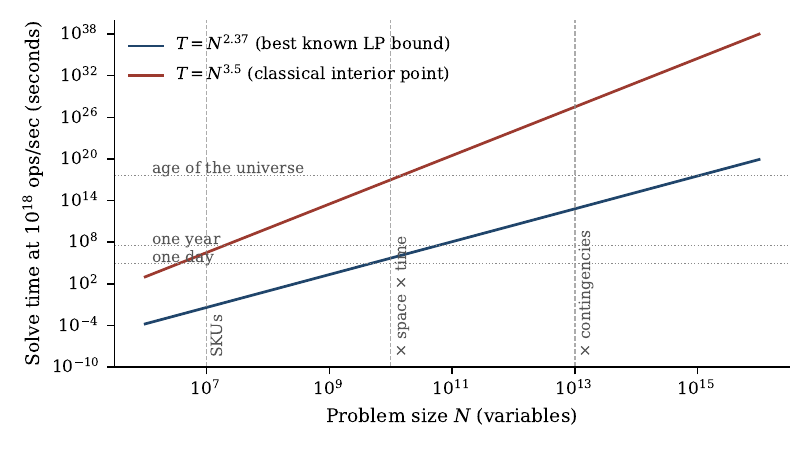}
\caption{\small \textbf{The computation horn, honestly drawn.} Solve time for a single deterministic program at $10^{18}$ operations per second, under the classical interior-point exponent and the best known theoretical bound. Vertical lines mark the growth of $N$ from bare commodity counts through space, time, and contingencies. The chart concedes the modern exponent; the binding constraints are nonconvexity (NP-completeness), the contingency dimension, and data staleness, none of which appear on the axis and all of which move the wall left.}
\label{fig:comp}
\end{figure}

\begin{corollary}[Goodhart--Lucas]\label{c:goodhart}
If the planner rewards an observable proxy $s$ linearly while true value has an unrewarded component, the optimum is a corner: all effort flows to the proxy. A measure that becomes a target ceases to measure.
\end{corollary}

The multitask formalization is Holmström and Milgrom (1991); the policy-invariance cousin is Lucas (1976). The Soviet archive supplies the comparative statics: reward tonnage and the factory casts one giant nail; reward chandelier weight and the ceilings come down.

\begin{proposition}[Soft budgets and queues]\label{p:soft}
(i) A producer rescued with probability $\pi$ chooses effort $e^{*}(\pi)$ decreasing in $\pi$; as $\pi\to 1$, $e^{*}\to e_{\min}$. A producer who cannot fail is a producer who need not produce (Kornai, 1980, 1986), and the rescue expectation is itself an equilibrium object, the center's sequentially rational refusal to cancel sunk projects (Dewatripont and Maskin, 1995). (ii) If the price is fixed at $\bar p<p^{*}$ and time is worth $v_z$ per unit, free entry into the queue drives the equilibrium wait to $w^{*}=(p^{*}-\bar p)/v_z$: the suppressed price difference is repaid in full, in time, to no one (Barzel, 1974).
\end{proposition}

Result (ii) deserves its bleak emphasis. The queue is a price. It is a price paid in the one currency that cannot be redistributed, collected by nobody, and it sorts by the shadow value of time rather than by willingness to pay, which is why the queue economy is not an egalitarian alternative to the price economy but a regressive implementation of it. The planned economies' joke literature understood this before their economists did. A useful boundary observation is Coase (1937): islands of planning inside firms are evidence not for planning but for its diminishing returns, since the firm expands exactly until internal planning costs meet external transaction costs. Walmart's interior is a plan; what disciplines Walmart is the exterior price.

\section{The Behavioral Layer}\label{sec:agent}

\begin{proposition}[The participation constraint]\label{p:part}
Under A6, an individual joins a revolt only if
\[
\beta_p\,\delta_p^{\tau}\; w(p)\; v(G) \;>\; \lambda\,(L+\kappa)^{\beta},
\]
where $G$ is the collective gain, $L$ what the individual loses under the status quo if the revolt is joined and fails or succeeds at her expense, $\kappa$ the participation cost, $\tau$ the delay, and $p$ compounds the probability of success with the probability that her participation is pivotal, which is $O(1/N)$ (Olson, 1965; Tullock, 1971). The pair $(\beta_p,\delta_p)$ denotes the quasi-hyperbolic present-bias and discount factors of A6; the subscript keeps them distinct from the value-function curvatures $\alpha,\beta$. Three discounts multiply the left side, distance, improbability, and the concave valuation of gains; one amplifier, $\lambda\approx 2.25$, multiplies the right.
\end{proposition}

The status quo, in this reading, is protected less by preference than by the architecture of preference. Prospect-theoretic treatments of rebellion exist (Berejikian, 1992; Levy, 1992), and relative deprivation is an older cousin (Gurr, 1970); what we add is the dynamics of the right-hand side, next.

\begin{proposition}[The J-window: threshold, opening, closing]\label{p:jwindow}
Under A6 and A7, suppose the economy has long delivered $c_t=\bar c$, so $r_{t_0}=\bar c$, and at $t_0$ consumption falls permanently to $\bar c-\Delta$. Then:
\begin{itemize}
\item[(i)] \emph{Threshold.} Suppose participation fails at the adapted point, $\beta_p\delta_p^{\tau}w(p)v(G)<\lambda(L+\kappa)^{\beta}$. Then there exists a unique $\bar z>0$ solving
\[
\lambda\Big[(\bar z+L+\kappa)^{\beta}-\bar z^{\beta}\Big] \;=\; \beta_p\,\delta_p^{\tau}\,w(p)\,G^{\alpha},
\]
and participation, the gap-$z$ analog of Proposition~\ref{p:part}'s comparison, becomes worthwhile exactly when the reference gap exceeds it, $|z_t|>\bar z$ (the failure branch is taken at weight one, a conservative simplification; see the proof). The perceptibility threshold is therefore not a free parameter: $\varepsilon\equiv\bar z$ is a derived object, increasing in the stake $L$, the cost $\kappa$, and loss aversion $\lambda$, and decreasing in the weighted prize $w(p)\,G^{\alpha}$.
\item[(ii)] \emph{Opening.} $z_{t_0+k}=c_{t_0+k}-r_{t_0+k}=-(1-\gamma)^{k}\Delta<0$: if $\Delta>\varepsilon$, the break itself opens the window, $|z_{t_0}|=\Delta>\varepsilon$, and revolt becomes individually rational. The mechanism is diminishing sensitivity: the incremental disvalue of staking $L+\kappa$, namely $\lambda[(|z|+L+\kappa)^{\beta}-|z|^{\beta}]$, falls toward zero as the ambient loss deepens.
\item[(iii)] \emph{Endogenous closing.} Absent further shocks, $|z_t|\leq\varepsilon$ after
\[
k^{*}=\frac{\ln(\Delta/\varepsilon)}{\ln\!\big(1/(1-\gamma)\big)}\;\approx\;\frac{1}{\gamma}\ln\frac{\Delta}{\varepsilon}
\]
periods: the reference adapts downward, the loss coding dissolves, and the constraint re-binds. Windows are transient, with length inversely proportional to the speed of hedonic adaptation; and breaks smaller than $\varepsilon$ open no window at all, which is why slow declines, accruing below threshold per period, produce no revolutionary moment however large their sum.
\end{itemize}
\end{proposition}

Part (ii) is Davies (1962) derived rather than hypothesized, and Tocqueville (1856) formalized: revolt arrives not at the bottom but at the break. Parts (i) and (iii) are, to our knowledge, new in this form: the threshold makes the window's edge a derived object rather than a free parameter, and the closing makes the window itself transient. The pair earns its keep quantitatively. Writing the half-life of adaptation as $h=\ln 2/\!\ln\!\big(1/(1-\gamma)\big)$, the window is $k^{*}=h\log_2(\Delta/\varepsilon)$. Panel-data estimates of adaptation to income put $h$ between about 1.4 and 3.1 years (Di Tella, Haisken-De New, and MacCulloch, 2010; Clark, Frijters, and Shields, 2008). For a shock eight times the perceptibility threshold, $\Delta/\varepsilon=8$, the window is four to nine years. Figure~\ref{fig:jwindow} displays the mechanism and the calibration.

Two caveats before the history. The calibration carries an adaptation speed estimated from subjective-well-being panels into the reference process that governs revolt participation. That these are one process is a bridging assumption: both are expressions of hedonic adaptation to material circumstances, which is why we regard the transfer as defensible, but it is a transfer, and the window lengths inherit its uncertainty. Second, the paragraph below is an illustration, not a test. We have made no census of windows that opened and closed without incident, so four celebrated successes cannot by themselves carry evidential weight against survivorship concerns; the systematic statement is prediction P3, which conditions on the window rather than sampling on the outcome.

The historical fit, for what illustration is worth, is not subtle. The revolutions that succeeded sit within half a decade of their precipitating breaks: 1786--88 fiscal collapse and 1789; the war losses beginning 1914 and 1917, a case where $\Delta$ kept growing and held the window open; 1977--78 and Iran 1979; the reform signals of 1988--89 and the autumn of 1989. The revolutions that were confidently expected and never came, in the long stagnations, are the cases where decline arrived slowly enough for the reference to track it, $\Delta$ accruing below $\varepsilon$ per period, no window ever opening. And the lament of every failed revolutionary generation, that the moment was missed, is not an excuse. It is the arithmetic of $\gamma$.

\begin{figure}[t]
\centering
\includegraphics[width=\textwidth]{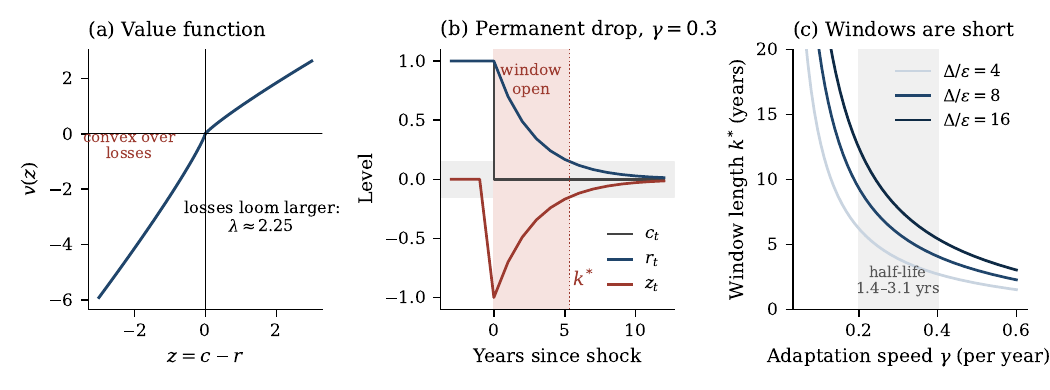}
\caption{\small \textbf{The J-window.} (a) The cumulative-prospect value function, loss aversion $\lambda\approx 2.25$, convex over losses. (b) A permanent one-unit consumption drop at $t=0$ with adaptation speed $\gamma=0.3$: the gap $z_t=c_t-r_t$ opens the window (shaded) and closes it at $k^{*}\approx 5.3$ years. (c) Window length $k^{*}=h\log_2(\Delta/\varepsilon)$ against $\gamma$; the shaded band is the adaptation speed implied by panel estimates in the happiness literature.}
\label{fig:jwindow}
\end{figure}

\begin{corollary}[Stakes]\label{c:stakes}
The constraint tightens in $L$. ``Nothing to lose but chains'' is the observation that the constraint loosens as $L\to 0$; the mortgage, the pension account, and the diploma are the observation that social policy raises $L$ across the base of the distribution, tightening the constraint by construction rather than by conspiracy.
\end{corollary}

\begin{corollary}[The lottery ticket]\label{c:lottery}
Overweighting of small probabilities, $w(p)>p$, subsidizes in psychological terms every within-system jackpot: the startup, the lottery, the viral post, the option grant. Mobility need not be common for the hope of mobility to be cheap to produce and expensive to abandon, and the system produces it in volume. A tension with Proposition~\ref{p:part} deserves resolution rather than a shrug: the same overweighting inflates the revolt gamble's $w(p)$, so which sign wins? The comparison decides, not the level. The within-system jackpot is a limited-liability bet, a small ticket with no $\lambda$-weighted stake; the revolt gamble carries $\lambda(L+\kappa)^{\beta}$ on its downside. Overweighting subsidizes both gambles; loss aversion taxes exactly one. The behavioral core does not suppress gambling, it channels it toward the gambles the system sells.
\end{corollary}

\begin{corollary}[The range of solidarity]\label{c:solidarity}
Repeated-game cooperation sustains itself where recognition and repetition are dense, small groups in the sense of Dunbar (1992), and hands off to anonymous mechanisms, price or coercion, beyond that range. The kibbutz working and Gosplan failing are the two faces of one folk theorem; the ultimatum-game evidence (Güth, Schmittberger, and Schwarze, 1982; Fehr and Gächter, 2002; Henrich et al., 2001) certifies both the reality of the fairness norm and its radius.
\end{corollary}

\begin{corollary}[The open-window share]\label{c:share}
If qualifying breaks ($\Delta>\varepsilon$) arrive as a Poisson process at rate $\mu$ and each opens a window of length $k^{*}$, the long-run fraction of time at least one window is open is $1-e^{-\mu k^{*}}$, the busy probability of an M/G/$\infty$ system. At one qualifying break per generation, $\mu=0.04$ per year, and $k^{*}=6$, about a fifth of history is window.
\end{corollary}

The corollary relocates the model's scarcity. What protects the regime is not the rarity of windows, which are open a nontrivial share of the time; it is the conjunction, inside a window, of a public-signal architecture and an inoperable dial. That conjunction is P3, and Section~\ref{sec:coord} is about its first leg.

\section{The Coordination Layer}\label{sec:coord}

An open window is necessary and not sufficient. Revolt is a coordination game, and which equilibrium gets played is a property of the information architecture, not of the grievance. We use the canonical regime-change game: a continuum of agents chooses attack or abstain at cost $t\in(0,1)$; the regime falls if the attacking mass reaches $\theta$; attackers receive 1 if it falls. Common knowledge is Aumann's (1976): an event is common knowledge at a state if the cell of the meet of the players' information partitions containing that state lies inside the event, so common knowledge of anything requires a \emph{public event} in which it can live.

\begin{proposition}[Coordination and the information architecture]\label{p:coord}
In the regime-change game:
\begin{itemize}
\item[(i)] Under complete information there are two equilibria, attack and no-attack, for all $\theta\in(0,1)$: coordination is possible everywhere in the interior, and selection is a matter of focal events (Schelling, 1960).
\item[(ii)] Under private noisy signals $x_i=\theta+\sigma\varepsilon_i$ only, the equilibrium is unique even as $\sigma\to 0$, and attack succeeds only for $\theta<\theta^{*}=1-t$ (Carlsson and van Damme, 1993; Morris and Shin, 1998, 2003): the entire interior band where coordination was possible under complete information goes quiet.
\item[(iii)] Adding a sufficiently precise public signal restores multiplicity, and with it the possibility of coordinated attack, disproportionately to its information content (Hellwig, 2002; Morris and Shin, 2002; Angeletos and Werning, 2006). No finite ladder of privately exchanged confirmations substitutes: with private channels, no finite order of mutual knowledge supports the risky action (Rubinstein, 1989), and even approximate common knowledge requires a public event to carry the common $p$-belief (Monderer and Samet, 1989).
\end{itemize}
Under A8, the personalized-media environment contains no public signal correlated with aggregate discontent; the cell of the meet containing the true state pools high- and low-discontent worlds, so ``the threshold has been crossed'' is common knowledge nowhere, and the system sits in regime (ii): unique no-attack for all but extreme fundamentals, however high the grievance.
\end{proposition}

Figure~\ref{fig:coord} draws the three regimes. The historical public-signal generators are exactly the objects the sociology of revolution keeps finding: the square, single-channel broadcasting, the Friday prayer network of 1978--79, the Leipzig Monday demonstrations whose informational cascade Lohmann (1994) documents week by week, the rituals whose function Chwe (2001) identifies as the manufacture of common knowledge. Preference falsification (Kuran, 1989, 1991) deepens the mechanism: under conformity pressure the fundamental $\theta$ is itself privately misperceived, which is why the cascades that do occur arrive, in Kuran's phrase, now out of never, and why regimes invest in information manipulation along precisely this channel (Edmond, 2013).

The reading of the present is then uncomfortable from both directions. The personalized feed dismantles the public-signal generators without a censor, by revealed preference, one perfectly tailored screen at a time. A million people holding the same grievance separately is not a movement; it is the definition of regime (ii). The point is not that discontent is low. The point is that the mapping from discontent to coordinated action runs through an architecture that no longer exists, and its disappearance required no one's decision.

\begin{remark}[Prices, virality, and the scope of A8]\label{r:scope8}
Four objections to A8 deserve answers on the page, two of them backed by evidence we would be foolish to wave away. First, the empirical literature runs against the naive sign: social-media penetration raised protest participation in Russia (Enikolopov, Makarin, and Petrova, 2020), and mobile coverage raised mobilization in Africa's downturns (Manacorda and Tesei, 2020). Both are treatments of \emph{connectivity}, access that lowers the cost of learning that others act. A8 is a claim about \emph{curation}, the replacement of commonly observed content by individually tailored feeds of privately unknown composition. Connectivity lowers coordination costs, which is their finding (their own mechanism evidence favors precisely this channel over the information channel); curation removes the layer on which common knowledge lives, which is our claim; both can be true at once, and the net sign is an empirical question, not a theorem. We state it as such, and P3 is its test. Second, the uniqueness we lean on is the static global game. In the dynamic version, learning from failed attacks restores multiplicity and generates waves of attacks with no exogenous public signal at all (Angeletos, Hellwig, and Pavan, 2007), and observed participation is itself an endogenous public signal, which is what Lohmann's Leipzig series in fact documents. Feeds make individual actions more visible, not less; what they obscure is the denominator, who else is seeing what I see. The honest statement is that the mechanism's sign is theoretically ambiguous and A8 takes a stand on it. Third, the market's own prices are public signals, and crashes have served as focal events; true, and Section~\ref{sec:adaptive} is the story of the state reading that signal faster than any crowd, the race of Figure~\ref{fig:record}b. Fourth, the 2019--24 wave, Hong Kong, Chile, Lebanon, Iraq, Sri Lanka, Bangladesh, is the natural stress test, and it delivered the signature the stand predicts: enormous synchronized sentiment; episodic mass presence where physical squares briefly reconstituted the public signal; regime change in exactly the two cases the jam indicator of Section~\ref{sec:predictions} scores as jammed, Sri Lanka in 2022 and Bangladesh in 2024; persistence wherever the dial could still turn, France's ten-billion-euro answer to the roundabouts of 2018--19 being the cleanest example; and no E2 anywhere. Lebanon, where window, square, and jam coincided without regime change, records that the conjunction is necessary in the data, not sufficient; the regime's raw strength $\theta$ is its own variable, and the model never promised sufficiency.
\end{remark}

\begin{figure}[t]
\centering
\includegraphics[width=0.9\textwidth]{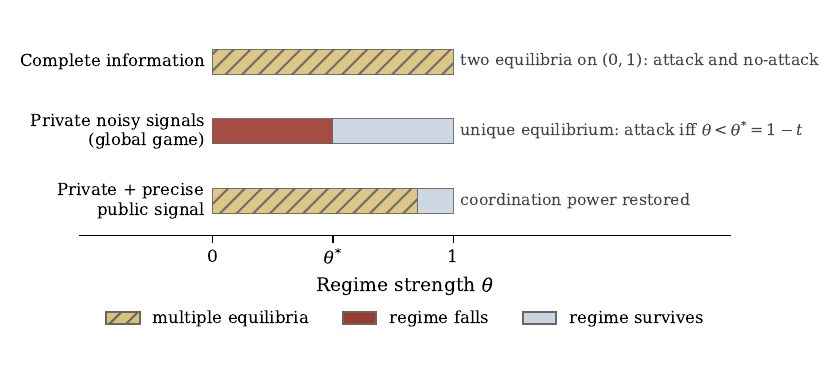}
\caption{\small \textbf{Three information regimes of the regime-change game.} Complete information admits both equilibria on the interior. Private noise selects the unique equilibrium with attack only below $\theta^{*}=1-t$ (drawn at $t=0.5$). A precise public signal restores coordination power. The passage from broadcast media and squares to personalized feeds is, in this model, a move from the third row to the second. The schematic depicts the static benchmark; the dynamic caveats are in Remark~\ref{r:scope8}.}
\label{fig:coord}
\end{figure}

\section{The Political Layer}\label{sec:political}

\begin{proposition}[The redistribution dial]\label{p:dial}
Under A10, in the elite--mass game where the elite chooses $\theta$ and the mass's outside option is the revolt whose value Propositions~\ref{p:part} and~\ref{p:jwindow} determine, the subgame-perfect equilibrium sets $\theta^{*}$ exactly at the binding participation constraint: revolution is never observed on the equilibrium path, and its threat disciplines $\theta$ continuously (Acemoglu and Robinson, 2000). Because the threat is transient, by Proposition~\ref{p:jwindow}(iii), promises made inside the window are not credible outside it; commitment is achieved by institutionalization, the franchise, constitutional social rights, automatic stabilizers (Acemoglu and Robinson, 2006). The welfare state is a commitment technology.
\end{proposition}

The empirical companion is direct: across Europe, 1820--1938, franchise extensions track revolution-threat indicators (Aidt and Jensen, 2014), and the prototype is Bismarck, who purchased the 1880s with social insurance while banning the party he was outbidding. The magnitude of the dial in modern data is the content of Figure~\ref{fig:gini}. In 2019, on total-population OECD figures, the United States, the United Kingdom, and Finland ran market-income Ginis of 0.505, 0.508, and 0.512, and disposable-income Ginis of 0.395, 0.366, and 0.273. One caution before the punchline: market Ginis on total populations are shaped by pension architecture and employment composition, since pensioners with near-zero market income inflate the statistic, a point Mogstad, Salvanes, and Torsvik (2025) press about Nordic comparisons in particular; the near-identity of the market numbers is a fact about the statistic, not proof of identical pre-tax economies, and we do not lean on it. What we lean on survives the caution untouched: twelve Gini points of disposable-income range inside one institutional family. Whatever else this establishes, it settles one definitional question: inequality at the American level is a policy setting within the system, not a property of it.

\begin{remark}[What binds on path, and what does not]\label{r:tension}
A tension in Proposition~\ref{p:dial} deserves daylight rather than a footnote. If the elite always sets the dial exactly at the binding constraint, revolution is off the equilibrium path, and a reader may conclude that the behavioral and coordination layers of Sections~\ref{sec:agent} and~\ref{sec:coord} do no equilibrium work. Half of that conclusion is wrong: the layers price the constraint, so $\lambda$, $\gamma$, and the derived $\varepsilon$ determine where the dial sits and how it moves, and the on-path observable is the dial's response to threat, which is prediction P7 and which Aidt and Jensen (2014) measure. The other half is right, and we concede it: P3, a claim about observed uprisings, cannot be an equilibrium prediction of a model in which uprisings never occur. Part of the repair already exists in the literature: Acemoglu and Robinson's dynamic model (2001) generates on-path transitions through transient windows and commitment failure, with revolt cheapest in recessions, a mechanism close in spirit to Proposition~\ref{p:jwindow}. What remains open, and what we propose, is elite uncertainty about the behavioral parameters specifically, the adaptation speed, the threshold, the current reference, under which the dial-setter faces an inference problem and revolts occur on path with positive probability. That model is the natural sequel to this paper, and it is not attempted here.
\end{remark}

\begin{figure}[t]
\centering
\includegraphics[width=0.82\textwidth]{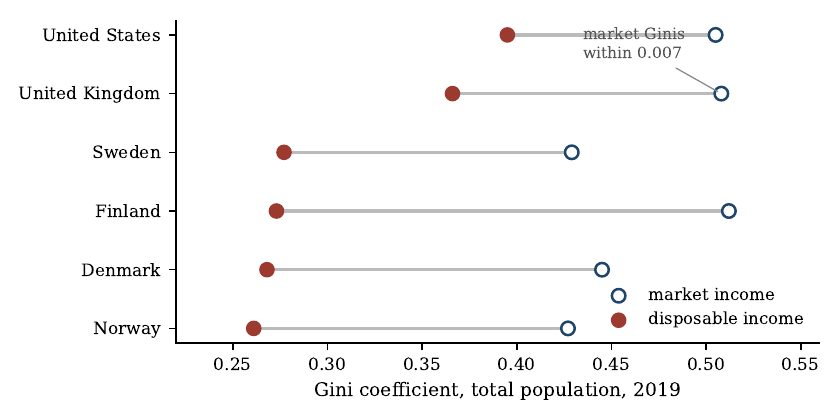}
\caption{\small \textbf{The dial is real and it is large.} Gini coefficients of market income (open) and disposable income (filled), total population, 2019, OECD Income Distribution Database as compiled by Mogstad, Salvanes, and Torsvik (2025). The United States, the United Kingdom, and Finland enter with market Ginis within 0.007 of one another and exit 12 points apart. Market Ginis on total populations are shaped by pension and employment composition; the paper's claim rests on the disposable spread, not on the market near-identity.}
\label{fig:gini}
\end{figure}

\begin{corollary}[Maintenance, not refutation]\label{c:piketty}
If $\theta$ is held fixed while the return on capital exceeds growth, the wealth share rises endogenously (Piketty, 2014). Within the model this is not a refutation of persistence; it is the control loop's maintenance manual. The dial does not hold its setting by itself; it is held there by political work, and Corollary~\ref{c:jam} is what happens when the work stops.
\end{corollary}

\begin{corollary}[Dial jams]\label{c:jam}
The historical cases where A10 fails, autocratic refusal in Russia through 1917, the deflationary vise of Germany 1929--33, are the cases where the model predicts instability rather than persistence. The prediction is directionless: a jammed dial opens onto 1917 and onto 1933. What the model forecasts at a jam is regime change, not its sign.
\end{corollary}

\section{Reinforcing Constraints}\label{sec:reinforce}

Three further mechanisms close escape routes that the first four leave notionally open. We state them briefly; proofs and sources are in the appendix.

\begin{proposition}[Conservation of status]\label{p:status}
(i) In one-dimensional costly-signaling games satisfying single crossing, separating equilibria exist and signaling expenditure is informative; the structure is isomorphic across the economic case (Spence, 1973) and the biological case (Zahavi, 1975), with evolutionary stability in the sense of Maynard Smith and Price (1973) established by Grafen (1990). (ii) Status is rank, and for any rank valuation $s(\cdot)$, $\sum_i s(\mathrm{rank}_i)$ is permutation-invariant: policy changes the carrier of rank, never its aggregate scarcity (Hirsch, 1976; Frank, 1985).
\end{proposition}

The demand for display is old (Veblen, 1899). Abolishing the price of status does not abolish status; it reroutes it, historically to party rank, metropolitan residence permits, and a black market in denim. The survey evidence that half of respondents prefer absolute poverty with higher rank (Solnick and Hemenway, 1998) says the demand side survives any supply-side reform.

\begin{proposition}[Feasible transitions only]\label{p:hotswap}
Under A9, every transformation that stays in $\mathcal{V}$ is a local-modular path, and shutting the system down for a rebuild is E1 by definition.
\end{proposition}

\begin{remark}[Fisher's geometry, as analogy]\label{r:fisher}
In Fisher's (1930) geometric model of adaptation, a random step of size $r$ taken at distance $d$ from the optimum in an $n$-dimensional trait space improves fitness with probability approximately $1-\Phi\!\big(r\sqrt{n}/(2d)\big)$, which vanishes as dimension grows for any fixed step-to-distance ratio; only small steps keep non-negligible odds (Orr, 1998). We import this as an analogy, not as a theorem about societies. What it disciplines is a habit of mind: in tightly coupled high-dimensional systems, the large random redesign is overwhelmingly likely to land below the incumbent, and the incremental step is the same geometry played to win. The revolutionary leap and the boring reform are bets on opposite ends of one distribution.
\end{remark}

Path stickiness is independently documented: the QWERTY keyboard (David, 1985) and increasing-returns lock-in (Arthur, 1989) show that history binds the feasible set even when no one is defending the incumbent. The corollary is the cleanest natural experiment the twentieth century ran, four times. Socialisms that reinstalled the price module survived, NEP 1921, Yugoslav market socialism, China 1978, Vietnam 1986; socialisms that did not, queued. Four cases, four reinstallations, one degree of freedom.

\begin{proposition}[The commons and the decarbonization margin]\label{p:commons}
(i) In an unpriced open-access resource, free entry proceeds until average product equals cost: overuse is the equilibrium of missing prices, not of private property (Gordon, 1954; Hardin, 1968). The Aral Sea was not privatized to death; it was unpriced to death, which is why abolishing private ownership universalizes rather than internalizes the externality. (ii) Small-scale, monitorable, repeated settings support self-governance, and Ostrom's (1990) design principles read as the field version of folk-theorem preconditions (laboratory confirmation in Ostrom, Walker, and Gardner, 1992); the same scale conditions say the solution does not lift to the planetary case. (iii) A price extends to an externality only if it covers it: carbon pricing is the missing line added to Lemma~\ref{l:prices}'s telegram, and Montreal 1987 is the proof of concept that a price-plus-rule mix can retire a global commons problem. (iv) Writing emissions as $E=Y\cdot I$ (the compressed Kaya identity; Kaya, 1990), the two margins are shrinking $Y$ and crashing $I$. Technology costs fall as a power law in cumulative production (Wright, 1936; validation in Nagy et al., 2013, and Farmer and Lafond, 2016), and the engine of cumulation is market scale. Figure~\ref{fig:solar} shows the flagship case: solar modules fell from \$106 per watt at 0.3 megawatts cumulative in 1976 to roughly \$0.11 at 1.6 terawatts in 2024, a decline of about 25 percent per doubling.
\end{proposition}

\begin{corollary}[The degrowth self-block]\label{c:degrowth}
The 1.5-degree pathway requires emission cuts of about 7.6 percent per year (UNEP, 2019). In 2020 the near-shutdown of the world economy delivered a fall of about 5.4 percent. The $Y$-margin therefore demands a 2020-scale contraction, made permanent and repeated annually, which by Proposition~\ref{p:jwindow} holds $z_t$ permanently negative and fires the $\lambda$-weighted loss response in every electorate every year. A program whose implementation requires continuously violating the participation constraint of its own electorate blocks itself; the $I$-margin, by contrast, runs on the very scaling mechanism of Proposition~\ref{p:commons}(iv). We engage the school's own framing, in which welfare is redefined off GDP, in Section~\ref{sec:predictions}; the participation constraint does not redefine with it.
\end{corollary}

\begin{figure}[t]
\centering
\includegraphics[width=0.68\textwidth]{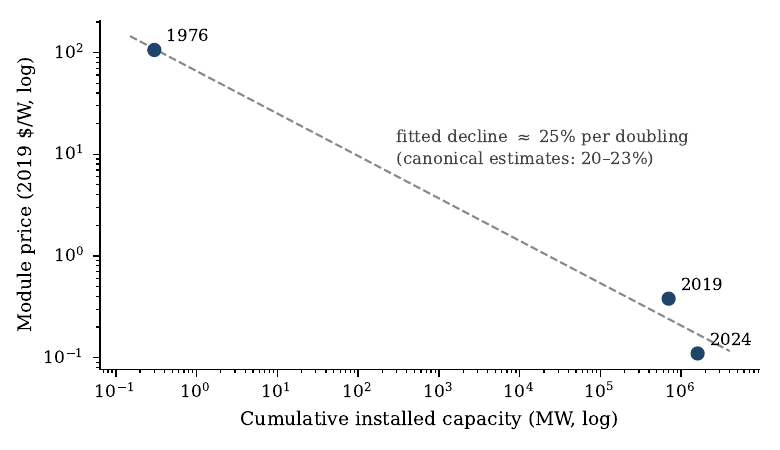}
\caption{\small \textbf{Wright's law at work.} Solar module prices against cumulative installed capacity, log-log. Anchors: \$106/W at 0.3 MW (1976), roughly \$0.38/W (2019) and \$0.11/W (2024). The fitted line implies a decline of roughly 25 percent per doubling; canonical estimates are 20--23 percent (Nemet, 2009; Farmer and Lafond, 2016; IRENA data via Roser, 2020).}
\label{fig:solar}
\end{figure}

\section{The Persistence Theorem}\label{sec:theorem}

\begin{theorem}[Persistence, by assembly]\label{t:main}
Under A1--A10:
\begin{itemize}
\item[(i)] $\Pj[\mathrm{E2}]=0$: no permanent non-price replacement of the price mechanism occurs.
\item[(ii)] $\Pj[\mathrm{E1}]\;\leq\;\sum_{k\in K}\dfrac{q_k^{(1)}}{1-\delta}\;+\;\Pj[\text{novel class}\ \wedge\ \text{adaptation overrun}]$, where the first term is fixed and small and the second is non-increasing in $t$ by A2.
\item[(iii)] The system's collapse probability is therefore a shrinking residual concentrated entirely in the novelty tail, and crisis is not the end of the regime but the input to its adaptation operator.
\end{itemize}
\end{theorem}

The proof, in the appendix, is an assembly job, which is the point: each joint is a separate literature's theorem. Replacement from above runs into Proposition~\ref{p:dilemma} and Lemma~\ref{l:prices}, and its interior decays by Corollary~\ref{c:goodhart} and Proposition~\ref{p:soft}. Replacement from below requires the conjunction of an agent (Propositions~\ref{p:part} and~\ref{p:jwindow}), a coordination architecture (Proposition~\ref{p:coord}), an opportunity (Corollary~\ref{c:jam}), and a feasible path (Proposition~\ref{p:hotswap}); and even when the conjunction obtains and a regime falls, the successor's permanent \emph{non-price} operation faces Propositions~\ref{p:dilemma} and~\ref{p:soft} afresh, which is what the four reinstallations of the price module were about. Disintegration is bounded by Proposition~\ref{p:immune} on known classes and by A2 on novel ones.

\begin{remark}[The theorem is exactly as strong as its assumptions, and no stronger in kind]\label{r:honest}
Two honesty clauses. The zero in part (i) is inherited from A4 and A5, no more and no less, and the tail in part (ii) is carried by A2. And the ten assumptions do not live in a single environment: Propositions~\ref{p:immune} and~\ref{p:anti} are hazard and portfolio calculations, Propositions~\ref{p:part} and~\ref{p:jwindow} are single-agent decision theory, Proposition~\ref{p:coord} is a continuum game, and Proposition~\ref{p:dial} is a two-player game with its own timing. Theorem~\ref{t:main} is a conjunction across models that do not nest; it organizes, it does not unify, and it inherits no equilibrium logic of its own. We therefore read it as a routing device that locates disagreement rather than a conclusion that closes it: objections belong at the assumptions, and Table~\ref{tab:rivals} shows where each school's objection in fact sits. The model that would unify the behavioral and political layers, an elite setting the dial under uncertainty about the mass's parameters, is the right next paper, not a paragraph of this one (Remark~\ref{r:tension}).
\end{remark}

\section{Predictions, Falsification, and Rival Positions}\label{sec:predictions}

\subsection{Testable predictions}

\begin{table}[t]
\centering
\small
\begin{tabular}{p{0.7cm} p{6.3cm} p{3.1cm} p{3.6cm}}
\toprule
 & \textbf{Prediction} & \textbf{Source} & \textbf{Evidence to date} \\
\midrule
P1 & Within a shock class, institutional response lag falls with each hit & A1--A2 & Fig.~\ref{fig:record}b: 816 $\to$ 18 $\to$ 2 days \\
P2 & Within a class, realized damage falls with each hit & Prop.~\ref{p:immune} & 1930s depositor losses vs.\ zero in 2023 \\
P3 & Successful mass uprisings cluster in (J-window) $\times$ (public signal) $\times$ (dial jam); grievance plus atomized media yields none & Props.~\ref{p:jwindow},~\ref{p:coord}; Cor.~\ref{c:jam} & 1789, 1917, 1979, 1989 vs.\ 2019--24: both regime changes under scored jams \\
P4 & Socialist regimes survive iff they reinstall the price module & Props.~\ref{p:dilemma},~\ref{p:hotswap} & NEP, Yugoslavia, China 1978, Vietnam 1986: 4/4 \\
P5 & Queue length and black-market premia comove with price-control coverage & Prop.~\ref{p:soft} & Poland 1980s; Venezuela 2016 \\
P6 & Explicit degrowth platforms face a structural vote ceiling & Cor.~\ref{c:degrowth} & no such platform has won a national majority \\
P7 & Redistribution expansions track revolution-threat indicators & Prop.~\ref{p:dial} & Aidt and Jensen (2014); Bismarck \\
\bottomrule
\end{tabular}
\caption{\small Seven predictions, their formal sources, and the current state of the evidence. P3's negative half is the sharpest near-term test: sustained high grievance under fully personalized media should keep producing large sentiment and small squares.}
\label{tab:predictions}
\end{table}

Table~\ref{tab:predictions} collects seven. Three comments. P3 is the one we would most like to see attacked, and for it to be attackable its jam indicator must be committed before outcomes, so here it is. A dial jam is scored when, during an open window, the transfer margin is blocked by at least one of: sovereign default or loss of market access; a binding external program constraining fiscal expansion, in force or entered within the window; or an institutional bar on the transfer instrument, a currency board, an occupation authority, a codified autocratic refusal. Scored on 2019--24, the indicator marks Sri Lanka 2022, Bangladesh 2024, and Lebanon 2020 as jams, and Hong Kong, Chile, Iraq, and France as unjammed; both regime changes in the set occurred under scored jams, none occurred without one, and Lebanon's jam without a regime change records that the condition is necessary, not sufficient (Remark~\ref{r:scope8}). One clean regime change through mass uprising without a scored jam refutes P3. Second, P3's negative half is being run as a natural experiment right now: grievance indices in several rich democracies sit at multi-decade highs, street mobilization is episodically enormous, and conversion into sustained coordination and regime outcomes stays confined to the jammed cases; the model reads this pattern as an information architecture at work, not a paradox. Third, P1 and P2 are the cheapest to test and the hardest to fake, since central banks and resolution authorities timestamp everything.

\subsection{Falsification conditions}

The framework is refuted by any of the following observations. The list is not closed; it is the frontier we can currently name.

\begin{itemize}
\item[F1.] \emph{An adaptation-outrunning shock class}: a class whose hazard does not fall with repetition, or whose first hit overruns the transfer capacity of A2. Feedback-driven climate tipping points are the live candidate; this is the theorem's largest known exposure, and we would rather name it than have it named for us.
\item[F2.] \emph{A preference-revelation technology}: a working mechanism that elicits true preferences and costs without binding budget constraints. This, and not faster hardware, is what an AI-planning research program must produce to defeat H2. We know of no candidate.
\item[F3.] \emph{A durable replacement}: a non-price replacement that runs a century without reinstalling prices and without disintegrating, however it arose (the horizon is fixed in Remark~\ref{r:op}); or any replacement arising in a prosperous society outside a J-window and without a public-signal generator. One clean case of either suffices.
\item[F4.] \emph{Dial capture}: a persistent non-response of $\theta$ to threat indicators, plutocratic lock-in. At capture, the model routes to Corollary~\ref{c:jam} and forecasts instability, so the persistence claim is withdrawn, not defended.
\item[F5.] \emph{The imported soft budget}: generalization of too-big-to-fail until Kornai's condition holds inside the system that once observed it from outside. Resolution regimes that fire management and wipe shareholders are the partial antidote; their erosion is the counter on this dial. The response-lag series of Figure~\ref{fig:record}b is this condition's own counter read from the other side: each faster backstop widens the guaranteed perimeter, so the paper's central evidence for A1 doubles as F5's needle, moving.
\item[F6.] \emph{Sensor erosion}: Proposition~\ref{p:prophecy}'s sign flip requires a continuing supply of credible collapse beliefs. If sustained success erodes the practice that produces the system's critics, the mechanism Lukyanov (2026) formalizes, then $g$ loses its decreasing slope and prophecy inversion decays. A measurable thinning of institutionalized skepticism alongside rising complacency in risk pricing is the counter.
\end{itemize}

\subsection{Where each rival school must attack}

\begin{table}[h]
\centering
\small
\begin{tabular}{p{4.2cm} p{3.3cm} p{6.4cm}}
\toprule
\textbf{Position} & \textbf{Must deny} & \textbf{Sharpest available test} \\
\midrule
Cumulative-crisis Marxism (Shaikh, 2016) & A1--A3 & P1--P2: are lags and damages rising or falling within class? \\
Market socialism (Lange, 1936--37) & scope of Lemma~\ref{l:prices}, Prop.~\ref{p:dilemma} & Lange keeps prices; the debate narrows to whether shadow prices solve H2 \\
Computational planning (Cockshott--Cottrell, 1993) & H1 & granted arguendo; H2 stands, F2 is the test \\
Degrowth (Kallis, 2018; Hickel, 2020) & framing of Cor.~\ref{c:degrowth} & P6, plus any electorate accepting a permanent 2020 \\
Capture thesis (Piketty, 2014) & durability of A10 & F4; inside the model it is maintenance, not refutation \\
\bottomrule
\end{tabular}
\caption{\small The routing table. Each school's disagreement with the theorem reduces to a named assumption and an observable test.}
\label{tab:rivals}
\end{table}

Table~\ref{tab:rivals} is, in our view, the paper's most useful export even for readers who reject its conclusion, because it converts a war of quotations into a schedule of measurements. One entry deserves a sentence of respect. The degrowth school's strongest reply is that welfare, correctly defined, detaches from $Y$, so Corollary~\ref{c:degrowth} prices the wrong object. Our response is empirical, not philosophical: the participation constraint of Proposition~\ref{p:part} is denominated in the electorate's own reference points, not in the analyst's welfare measure, and no electorate to date has voted as if the redefinition had occurred. P6 keeps score.

\subsection{Scope}

Three limits. The theorem is positive, not normative; persistence is compatible with, and in Proposition~\ref{p:status} and Corollary~\ref{c:piketty} coexists with, tensions the model makes no attempt to bless. The assumptions describe a historical-institutional envelope; if A8's information structure or A10's control loop is rewritten by technology or politics, the conclusions move with the envelope, which is what the falsification list is for. And the model's objects, $\mathcal{V}$, $q_k^{(m)}$, $\gamma$, are not directly observed; the predictions of Table~\ref{tab:predictions} are their observable shadows, and the tests live there.

\section{Conclusion}\label{sec:conclusion}

We asked why a system whose death has been scheduled since 1848 keeps missing the appointment, and we answered with a machine of four interlocking parts: a rulebook that metabolizes each crisis into code, at a speed that has fallen from years to a weekend within the best-documented shock class; a communication technology, the price, whose centralized replacement fails on the input horn even after every concession to modern computation; an agent whose participation in revolt is confined to windows that hedonic adaptation closes in roughly $h\log_2(\Delta/\varepsilon)$ years; and a coordination architecture in which the disappearance of public signals parks protest in the unique quiet equilibrium of a global game, with a distribution dial turned, by a threatened elite's own self-interest, exactly far enough to keep the window from mattering.

None of this makes the system just, and the theorem says nothing about deserving. It says the system is stable, states exactly the ten claims that stability rests on, and posts the six ways the claims can fail, with climate feedbacks and the creeping soft budget at the head of the list. The 175-year forecasting record, on this reading, is not a mystery and not a taunt. It is what an adaptive adversary's ledger looks like from the outside: every prediction of its death filed, read, and patched. The most durable contribution of the system's critics has been involuntary, and it is measurable in Figure~\ref{fig:record}. They are its immune system's antigen supply.

\newpage
\appendix

\section{Proofs}\label{app:proofs}

\begin{proof}[Proof of Proposition~\ref{p:immune}]
(i) Iterating A1 gives $q_k^{(m)}\leq\delta^{m-1}q_k^{(1)}$; sum the geometric series. (ii) Collapse is an absorbing event, so ``collapse ever occurs from a known class'' is the countable union $\bigcup_{k}\bigcup_m C_k^{(m)}$ of single-hit collapse events, and $\Pj[\bigcup A]\leq\sum\Pj[A]$ requires no independence. (iii) Since $\sum_m p_k^{(m)}\leq p_k^{(1)}/(1-\delta)<\infty$, the first Borel--Cantelli lemma, which likewise requires no independence, gives $\Pj[\text{infinitely many class-}k\text{ near misses}]=0$.
\end{proof}

\begin{proof}[Proof of Proposition~\ref{p:anti}]
$z\mapsto\max\{z,0\}$ is convex; by A3, $V_i(x+\epsilon)$ is affine in $\epsilon$; the composition of a convex function with an affine map is convex, and sums of convex functions are convex, so $F(x+\epsilon)$ is convex in $\epsilon$. Jensen's inequality gives $\E[F(x+\epsilon)]\geq F(x+\E\epsilon)=F(x)$. Strictness follows from the strict convexity of $\max\{\cdot,0\}$ in any neighborhood straddling a kink that carries positive probability on both sides.
\end{proof}

\begin{proof}[Proof of Proposition~\ref{p:prophecy}]
With $g$ continuous and strictly decreasing on $[0,1]$, $h(q)=g(q)-q$ is continuous and strictly decreasing with $h(0)\geq 0\geq h(1)$; the intermediate value theorem and monotonicity give a unique root $q^{*}$. Under $\sup|g'|<1$ the map is a contraction and Banach's fixed-point theorem delivers convergence; for steeper $g$ the familiar cobweb two-cycle is possible, which is why the contraction condition appears in the statement.
\end{proof}

\begin{proof}[Proof of Proposition~\ref{p:dilemma}]
\emph{H1.} Suppose first the problem is convex, contrary to A5, and deterministic with $N$ variables. Interior-point methods solve LPs in $O(N^{3.5}L)$ arithmetic operations; the best known theoretical bound is $O^{*}(N^{\omega})$ with $\omega\approx 2.37$ (Cohen, Lee, and Song, 2019). At $N=10^{10}$ and $10^{18}$ operations per second these give roughly $10^{17}$ seconds and $10^{5.7}$ seconds (about six days) per solve, respectively. Contingent commodities multiply $N$ by the number of distinguishable states, which grows combinatorially in the number of elementary uncertain events, and each re-solve must complete within the staleness horizon $\tau_{s}$ of A5. Dropping the counterfactual convexity: with indivisibilities and increasing returns the feasible set is nonconvex and the problem contains integer programming, which is NP-complete (Karp, 1972); under $\mathrm{P}\neq\mathrm{NP}$ no algorithm solves all instances in time polynomial in $N$. \emph{H2.} If reports carry no cost, the direct mechanism is a cheap-talk game in which the allocation is monotone in reported need; truthful separation is not incentive compatible, and the babbling equilibrium survives. If reports carry a cost that scales with the allocation, the mechanism has reinvented a price. The general obstructions: in classical environments no Pareto-efficient, individually rational mechanism is dominant-strategy incentive compatible (Hurwicz, 1972); and even bilateral trade admits no mechanism that is Bayesian incentive compatible, interim individually rational, budget balanced, and ex-post efficient (Myerson and Satterthwaite, 1983). By A4(ii), part of the input does not exist as reportable knowledge prior to choice under a binding constraint, so no report-based mechanism has access to it even in principle.
\end{proof}

\begin{proof}[Proof of Corollary~\ref{c:goodhart}]
With payoff $\omega s(e_s)-c(e_s+e_u)$ and true value $u(e_u)$ unrewarded, any interior allocation of effort to $e_u$ is dominated by transferring it to $e_s$; the optimum is the corner $e_u^{*}=0$. The multitask generalization, with risk and complementarities, is Holmström and Milgrom (1991).
\end{proof}

\begin{proof}[Proof of Proposition~\ref{p:soft}]
(i) The producer solves $\max_e R(e)-c(e)-(1-\pi)D(e)$ with $c$ convex increasing and $D$ the bankruptcy loss, decreasing in effort. The first-order condition $c'(e)=R'(e)-(1-\pi)D'(e)$ has $e^{*}$ decreasing in $\pi$ since the marginal deterrence term $-(1-\pi)D'(e)>0$ shrinks as $\pi$ rises; at $\pi=1$ only $R'$ disciplines effort, and where $R$ is administratively fixed, $e^{*}=e_{\min}$. Dewatripont and Maskin (1995) endogenize $\pi$: with sunk project costs, refinancing dominates cancellation ex post, so a soft budget is the center's sequentially rational policy and is anticipated. (ii) With unit demand, valuations above $p^{*}$, and homogeneous time value $v_z$, entry into the queue continues until the marginal entrant's total price $\bar p+v_z w$ equals the market-clearing price $p^{*}$, giving $w^{*}=(p^{*}-\bar p)/v_z$; heterogeneous $v_z$ sorts entrants by time cost (Barzel, 1974).
\end{proof}

\begin{proof}[Proof of Proposition~\ref{p:part}]
Direct application of A6 to the join/abstain gamble. Joining pays $G$ with probability $p$ at delay $\tau$ and costs $\kappa$ now, with $L$ at risk; abstention preserves the status quo at reference. Discounting, probability weighting, and the value function map these into the displayed inequality. That the individual treats $p$ as exogenous follows from $\partial p/\partial(\text{own participation})=O(1/N)$ (Olson, 1965; Tullock, 1971).
\end{proof}

\begin{proof}[Proof of Proposition~\ref{p:jwindow}]
(i) Write the join/abstain comparison at ambient gap $z<0$, taking the failure branch at weight one (conservative, since $w(p)$ is small): abstaining continues the status quo at $v(z)=-\lambda|z|^{\beta}$; joining pays $\kappa$ with stake $L$, reaching $-\lambda(|z|+L+\kappa)^{\beta}$ on failure and, with weight $\beta_p\delta_p^{\tau}w(p)$, the restored-reference prize $G^{\alpha}$ on success. Participation is worthwhile iff $\beta_p\delta_p^{\tau}w(p)G^{\alpha}>\lambda[(|z|+L+\kappa)^{\beta}-|z|^{\beta}]$. The right side is continuous and strictly decreasing in $|z|$ on $(0,\infty)$, since its derivative is $\lambda\beta[(|z|+L+\kappa)^{\beta-1}-|z|^{\beta-1}]<0$ for $\beta<1$; it runs from $\lambda(L+\kappa)^{\beta}$ at $|z|=0$ toward $0$. The adapted-point failure assumption places the constant left side strictly below the right side's value at zero, so the intermediate value theorem and monotonicity deliver a unique crossing $\bar z>0$, with participation exactly for $|z|>\bar z$. Monotone comparative statics in $(\lambda,L,\kappa)$ and in $w(p)G^{\alpha}$ follow from the same monotonicity; restoring the failure branch to weight $1-w(p)$ shifts $\bar z$ continuously and changes nothing qualitative. (ii) The reference dynamic is linear: for $k\geq 0$, $r_{t_0+k}=(\bar c-\Delta)+(1-\gamma)^{k}\Delta$, so $z_{t_0+k}=-(1-\gamma)^{k}\Delta$ and $|z_{t_0}|=\Delta>\varepsilon$ opens the window at the break. (iii) Setting $|z_t|\leq\varepsilon$ and solving for $k$ gives the display; the approximation uses $\ln(1/(1-\gamma))\approx\gamma$ for small $\gamma$. For Corollary~\ref{c:share}: windows generated by a Poisson process of rate $\mu$, each of length $k^{*}$, form an M/G/$\infty$ system whose stationary probability of at least one busy server is $1-e^{-\mu k^{*}}$.
\end{proof}

\begin{proof}[Proof sketch of Proposition~\ref{p:coord}]
(i) For $\theta\in(0,1)$, all-attack and no-attack are both Nash under complete information. (ii) In the switching equilibrium of the global game the marginal agent, observing the cutoff signal, holds (in the small-noise limit with diffuse prior) a uniform belief over the proportion of others attacking; indifference requires $\Pj[\text{success}]=1-\theta^{*}=t$, giving $\theta^{*}=1-t$, and iterated dominance makes the switching profile the unique equilibrium (Carlsson and van Damme, 1993; Morris and Shin, 1998, 2003). (iii) Sufficiently precise public signals concentrate higher-order beliefs and restore multiplicity (Hellwig, 2002; Morris and Shin, 2002; Angeletos and Werning, 2006); the email-game induction of Rubinstein (1989) shows that no finite chain of private confirmations makes the risky action a best reply, and Monderer and Samet (1989) show that the common $p$-belief that would suffice must itself live on a public event. The final claim applies Aumann's characterization: if no event measurable in every partition separates high- from low-discontent states, the meet's cell containing the truth pools them, and the threshold event is common knowledge at no state.
\end{proof}

\begin{proof}[Proof of Proposition~\ref{p:dial}]
Backward induction: if the constraint is slack the elite lowers $\theta$ (redistribution is costly to it); if violated, revolt follows; the optimum binds the constraint. Transience of the threat makes any promise of future $\theta$ inside the window non-credible outside it; transferring formal decision rights (the franchise) converts the promise into a state variable, which is the commitment mechanism of Acemoglu and Robinson (2006).
\end{proof}

\begin{proof}[Proof of Proposition~\ref{p:status}]
(i) Standard: under single crossing the high type's signal is cheaper at the margin, supporting separation; Grafen (1990) establishes the evolutionary version as an ESS. (ii) Rank is a permutation of the population; $\sum_i s(\mathrm{rank}_i)=\sum_{j=1}^n s(j)$ for any permutation.
\end{proof}

\begin{proof}[Proof of Proposition~\ref{p:hotswap}]
Immediate from Definition~\ref{d:collapse} and A9; the Fisher--Orr geometry invoked alongside it is an analogy and is labeled as such in Remark~\ref{r:fisher}.
\end{proof}

\begin{proof}[Proof of Proposition~\ref{p:commons} and Corollary~\ref{c:degrowth}]
(i) Open access equates average product to cost (Gordon, 1954). (iii) is definitional given Lemma~\ref{l:prices}'s scope. (iv) Wright's law is the empirical regularity $\log c = a - b\log Q$ across the technology panel of Nagy et al.\ (2013). For the corollary: sustaining $E$ cuts of 7.6 percent per year on the $Y$-margin with $I$ fixed requires $\Delta Y/Y\approx-7.6$ percent annually; 2020's realized $-5.4$ percent emission fall accompanied the sharpest output contraction of the postwar era, so the required policy holds $z_t<-\varepsilon$ permanently, and Proposition~\ref{p:jwindow}(ii) prices the electoral response at $\lambda$ per unit.
\end{proof}

\begin{proof}[Proof of Theorem~\ref{t:main}]
(i) An E2 event is either designed or fought for. The designed route contradicts Proposition~\ref{p:dilemma}: any non-price mechanism hits H1 or H2 within each staleness horizon, and its interior decays by Corollary~\ref{c:goodhart} and Proposition~\ref{p:soft}. The fought-for route requires the conjunction of Propositions~\ref{p:part} and~\ref{p:jwindow} (an agent, inside a window), Proposition~\ref{p:coord} (a public-signal architecture), Corollary~\ref{c:jam} (a jammed dial; otherwise Proposition~\ref{p:dial} clears the constraint), and Proposition~\ref{p:hotswap} (a feasible path). Even on that measure-theoretically thin conjunction, what follows is regime change; the successor's \emph{permanent non-price operation} re-enters the designed route and fails there; historically, this is what the four reinstallations of the price module recorded. Hence the E2 set is empty under the assumptions. (ii) Known classes contribute at most $\sum_k q_k^{(1)}/(1-\delta)$ by Proposition~\ref{p:immune}; novel classes contribute at most the A2 tail, non-increasing by meta-learning. (iii) is (i) plus (ii) plus Definition~\ref{d:collapse}.
\end{proof}

\section{Data and Figure Construction}\label{app:data}

\textbf{Figure~\ref{fig:record}a.} Contraction durations in months, peak to trough, NBER Business Cycle Dating Committee chronology, all cycles with peaks 1873--2020. \textbf{Figure~\ref{fig:record}b.} Lag definitions: 1930--33, from the failure of the Bank of United States (December 11, 1930), conventionally the start of the systemic phase of the first banking panic, to the national bank holiday of March 6, 1933 (the Emergency Banking Act followed on March 9 and the FDIC in June): 816 days. 2008: Lehman Brothers failure (September 15) to the Emergency Economic Stabilization Act (October 3): 18 days; the Treasury money-market guarantee came on September 19, 4 days. 2023: closure of Silicon Valley Bank (Friday, March 10) to the joint announcement of the systemic-risk exception and the Bank Term Funding Program (Sunday, March 12): 2 days. Severity companion: roughly nine thousand bank suspensions and about \$1.3 billion of cumulative depositor losses in 1930--33, against zero depositor losses in 2023. NBER dates before 1927 are retrospective Burns--Mitchell reconstructions and are not strictly comparable to the modern committee's. \textbf{Table~\ref{tab:record}.} Peak disruption entries: NBER durations; Black Monday, $-22.6$ percent (DJIA, October 19, 1987); Nasdaq composite peak-to-trough 2000--02, approximately $-78$ percent; U.S. real output decline 1929--33, roughly one quarter.

\textbf{Figure~\ref{fig:jwindow}.} Panel (b): $\gamma=0.3$, $\Delta=1$, $\varepsilon=0.15$, giving $k^{*}=\ln(1/0.15)/\ln(1/0.7)\approx 5.3$. Panel (c): $k^{*}=h\log_2(\Delta/\varepsilon)$ with the shaded band $\gamma\in[0.2,0.4]$, equivalently half-lives $h\in[1.4,3.1]$ years, the range implied by adaptation estimates in Di Tella, Haisken-De New, and MacCulloch (2010) and the survey evidence in Clark, Frijters, and Shields (2008). With Proposition~\ref{p:jwindow}(i), $\varepsilon$ is the derived threshold $\bar z$, so $\Delta/\varepsilon$ reads as the ratio of the break to that object.

\textbf{Figure~\ref{fig:gini}.} Gini coefficients of market and disposable income, total population, 2019, OECD Income Distribution Database, as compiled in Mogstad, Salvanes, and Torsvik (2025): United States 0.505/0.395; United Kingdom 0.508/0.366; Finland 0.512/0.273; Denmark 0.445/0.268; Sweden 0.429/0.277; Norway 0.427/0.261.

\textbf{Figure~\ref{fig:solar}.} Module price anchors, 2019 dollars per watt against cumulative installed capacity: \$106 at 0.3 MW (1976); approximately \$0.38 around 700 GW (2019); approximately \$0.11 at 1.6 TW (2024). Sources: Nemet (2009), Farmer and Lafond (2016), IRENA series as presented in Roser (2020) and subsequent updates. The fitted per-doubling decline of about 25 percent brackets the canonical 20--23 percent range.

\textbf{Other quantities.} Emissions: required cut of 7.6 percent per year for the 1.5-degree pathway, UNEP (2019); realized 2020 fall in fossil CO$_2$ emissions, approximately 5.4 percent (Global Carbon Project). Suez blockage: approximately \$9.6 billion per day of trade flow held up (Lloyd's List estimate), a delay figure, not a destruction figure. Gosplan line items: roughly twelve million (Nove, 1983). The binomial bound in Section~\ref{sec:record}: $1-0.05^{1/15}\approx 0.181$.

\section{Notation}\label{app:notation}

\begin{table}[h]
\centering
\small
\begin{tabular}{l p{10.6cm}}
\toprule
\textbf{Symbol} & \textbf{Meaning} \\
\midrule
$E_t=(x_t,R_t)$ & state of the economy; rulebook \\
$\mathcal{V},\ \mathcal{V}_\varepsilon$ & viability set; stressed neighborhood \\
$\Sigma_k,\ N_k(t)$ & shock class $k$; cumulative hits \\
$q_k^{(m)},\ p_k^{(m)}$ & collapse and near-miss hazard of the $m$-th hit \\
$\delta$ & immune-memory decay factor (A1) \\
$A(\cdot)$ & adaptation operator \\
$F(x)=\sum_i\max\{V_i,0\}$ & option-structured aggregate value \\
$v(\cdot),\lambda,\alpha,\beta,w(p)$ & CPT value function, loss aversion, curvatures, probability weight \\
$\beta_p,\ \delta_p$ & quasi-hyperbolic present-bias and discount factors (distinct from the curvatures $\alpha,\beta$) \\
$r_t,\ \gamma,\ z_t=c_t-r_t$ & reference point; adaptation speed; reference-relative gain \\
$h=\ln 2/\ln\!\big(1/(1-\gamma)\big)$ & adaptation half-life \\
$k^{*}$ & J-window length \\
$\varepsilon\equiv\bar z$ & derived perceptibility threshold (Proposition~\ref{p:jwindow}(i)) \\
$\mu$ & arrival rate of qualifying breaks \\
$G, L, \kappa, \tau$ & revolt gain; status-quo stake; participation cost; delay \\
$\theta,\ \theta^{*},\ t$ & regime strength; global-game cutoff; attack cost \\
$N,\ \tau_{s}$ & planning problem size; data staleness horizon (distinct from the participation delay $\tau$) \\
$\pi,\ w^{*}$ & bailout probability; equilibrium queue wait \\
\bottomrule
\end{tabular}
\caption{\small Notation.}
\end{table}

\end{document}